\documentclass[12pt]{amsart}
\usepackage{amssymb}
\usepackage{color}
\newtheorem{claim}{}[section]
\newtheorem{theorem}[claim]{Theorem}

\renewenvironment{proof}{\noindent{\it Proof. \hskip0pt}}
                      {$\square$\par\medskip}

\begin{document}
\baselineskip 6.2 truemm
\parindent 1.5 true pc

\newcommand\lan{\langle}
\newcommand\ran{\rangle}
\newcommand\tr{{\text{\rm Tr}}\,}
\newcommand\ot{\otimes}
\newcommand\wt{\widetilde}
\newcommand\join{\vee}
\newcommand\meet{\wedge}
\renewcommand\ker{{\text{\rm Ker}}\,}
\newcommand\im{{\text{\rm Im}}\,}
\newcommand\mc{\mathcal}
\newcommand\transpose{{\text{\rm t}}}
\newcommand\FP{{\mathcal F}({\mathcal P}_n)}
\newcommand\ol{\overline}
\newcommand\JF{{\mathcal J}_{\mathcal F}}
\newcommand\FPtwo{{\mathcal F}({\mathcal P}_2)}
\newcommand\hada{\circledcirc}
\newcommand\id{{\text{\rm id}}}
\newcommand\tp{{\text{\rm tp}}}
\newcommand\pr{\prime}
\newcommand\e{\epsilon}
\newcommand\inte{{\text{\rm int}}\,}
\newcommand\ttt{{\text{\rm t}}}
\newcommand\spa{{\text{\rm span}}\,}
\newcommand\conv{{\text{\rm conv}}\,}
\newcommand\rank{\ {\text{\rm rank of}}\ }
\newcommand\vvv{\mathbb V_{m\meet n}\cap\mathbb V^{m\meet n}}
\newcommand\ppp{\mathbb P_{m\meet n} + \mathbb P^{m\meet n}}
\newcommand\re{{\text{\rm Re}}\,}
\newcommand\la{\lambda}
\newcommand\msp{\hskip 2pt}
\newcommand\ppt{\mathbb T}
\newcommand\rk{{\text{\rm rank}}\,}
\def\cc{\mathbb{C} }
\def\pp{\mathbb{P} }
\def\rr{\mathbb{R} }
\def\qq{\mathbb{Q} }
\def\a{\alpha }
\def\b{\beta }

\title{Necessary conditions for optimality of decomposable entanglement witnesses}

\author{Seung-Hyeok Kye}
\address{Department of Mathematics and Institute of Mathematics\\Seoul National University\\Seoul 151-742, Korea}
\email{kye@snu.ac.kr}

\thanks{partially supported by NRFK 2011-0001250}



\begin{abstract}
     It is well known that the support of an optimal decomposable
entanglement witness is completely entangled. We add two more
necessary conditions for the optimality: The orthogonal complement
of the support must have a nonzero product vector; another one will
be given in terms of related faces of a convex cone. With these
necessary conditions, we show that there exist examples of
non-optimal decomposable entanglement witnesses which are the
partial transposes of positive semi-definite matrices supported on
completely entangled spaces, whenever both of the local dimensions
are greater than or equal to three.\end{abstract}

\maketitle

\section{Introduction}

Denote by $M_n$ the $C^*$-algebra of all $n\times n$ matrices over the complex field.
A positive semi-definite $mn\times mn$ matrix in $M_{mn}=M_n \otimes M_m$
is said to be separable if it is the convex sum of rank one positive semi-definite
matrices onto product vectors in $\mathbb C^n\otimes \mathbb C^m$, where a product vector
is nothing but a simple tensor of the form $x\otimes y\in \mathbb C^n\otimes \mathbb C^m$.
We denote by $\mathbb V_1$ the cone of all separable ones.
A positive semi-definite matrix in $M_n\otimes M_m$
is said to be entangled if it is not separable.
The notion of entanglement is a key research area of quantum physics
in relation with possible applications to quantum information and quantum computation theory.
See \cite{guhne} and \cite{horo-survey} for related survey articles.

An entanglement witness \cite{terhal} is a block-positive matrix $W$ in $M_n\otimes
M_m$ which detects entanglement. It is said to be optimal in \cite{lew00} if it
detects a maximal set of entanglement with respect to the set inclusion.
In particular, it was proven in \cite{lew00} that if a decomposable witness is optimal
then it takes the form of the partial transpose $Q^\tau$ of
a positive semi-definite matrix $Q$. Consider the following two
conditions on the support $E$ of $Q$:
\begin{enumerate}
\item[(A)]
The partial conjugates of the product vectors in $E^\perp$ span the whole space $\mathbb C^n\otimes \mathbb C^m$.
\item[(B)]
$E$ is completely entangled, that is, there is no product vector in $E$.
\end{enumerate}
It was shown in \cite{lew00} that the condition (A) implies the optimality of $W=Q^\tau$, and the optimality of $Q^\tau$
implies the condition (B).

In the case that one of the local dimensions is $2$, it was shown in \cite{aug} that both of the above two conditions
turn out to be necessary and sufficient conditions for the optimality of $W=Q^\tau$. In general cases, it was also shown in \cite{asl} that
the condition (A) is not necessary for the optimality of $W=Q^\tau$ whenever both of the local dimensions are equal to
or greater than $3$.

The purpose of this note is to show that the condition (B) is not
sufficient for the optimality of $W=Q^\tau$, whenever both of the
local dimensions are equal to or greater than $3$. To do this, we
give two additional necessary conditions for the optimality of
$W=Q^\tau$. If $Q^\tau$ is an optimal decomposable witness then we
have the following:
\begin{enumerate}
\item[(C)]
The orthogonal complement of the support of $Q$ must have a nonzero product vector.
\item[(D)]
The convex cone consisting
of all positive semi-definite matrices supported on the support of $Q$ must be a face of the convex cone
generated by all positive semi-definite matrices and their partial transposes.
\end{enumerate}
With these additional necessary conditions,
it turns out that there exist examples of non-optimal entanglement witnesses $W=Q^\tau$
with the condition (B), unless one of the local dimensions is $2$.

By the Jamio\l kowski-Choi isomorphism \cite{choi75-10}, \cite{jami}, block-positive matrices in $M_n\otimes M_m$ correspond to positive
linear maps from $M_m$ into $M_n$. Since our result heavily depends on the boundary structures of the convex cone
consisting of all positive linear maps, we will call sometimes a positive linear map itself an entanglement witness
whenever it has the corresponding property under this isomorphism, as in the recent paper \cite{ha+kye_exposed}, where
indecomposable entanglement witnesses arising from positive linear maps generating exposed rays are considered.

In the next section, we examine the duality between entanglement and positive linear maps to get the above
necessary conditions (C) and (D). During the discussion, the above mentioned necessary condition (B) for the optimality
will be recovered in the framework of the facial structures.
In the last section, we present examples and discuss the related topics.

\section{Results}

To determine if a given positive semi-definite matrix is separable, the duality
between the tensor product $M_{mn}=M_n \otimes M_m$ and the space ${\mathcal L}(M_m,M_n)$ of
all linear maps from $M_m$ into $M_n$ plays a key role.
For $x\in M_m$, $y\in M_n$ and $\phi\in{\mathcal L}(M_m,M_n)$, we define the bilinear pairing by
\begin{equation}\label{ddual}
\langle y\otimes x, \phi \rangle =\tr (\phi(x)y^\ttt),
\end{equation}
where $y^\ttt$ denotes the transpose of $y$.
Then the cone $\mathbb V_1$ and the cone $\mathbb P_1$ consisting of all positive maps in
${\mathcal L}(M_m,M_n)$ are dual to each other \cite{eom-kye} in the following sense:
\begin{eqnarray}
A\in\mathbb V_1\ &\Longleftrightarrow\ \langle A,\phi\rangle \ge 0\ {{\rm for\ each}}\ \phi\in\mathbb P_1,\\
\phi\in\mathbb P_1\ &\Longleftrightarrow\ \langle A,\phi\rangle \ge 0\ {{\rm for\ each}}\ A\in\mathbb V_1.
\end{eqnarray}
Therefore, every entangled state is detected by
a positive map in the following sense: A positive semi-definite matrix
$A$ in $M_n\otimes M_m$ is entangled if and only if
there is a positive map $\phi\in\mathbb P_1$ such that $\langle A,\phi\rangle <0$.
This is an equivalent formulation
to the separability criterion given in \cite{horo-1} under the Jamio\l kowski-Choi isomorphism.

Basic examples of positive maps from $M_m$ into $M_n$ come from
elementary operators together with the transpose map:
$$
\phi_V: X\mapsto V^*XV,
\qquad
\phi^V: X\mapsto V^*X^\ttt V,
$$
where $V$ is an $m\times n$ matrix.
The convex sums of the first (respectively second) types are said to be completely positive (respectively
completely copositive) linear maps, and
the convex sums of completely positive maps and
completely copositive maps are said to be decomposable positive maps.
We denote by $\mathbb D$ the cone of all decomposable positive maps.

Under the Jamio\l kowski-Choi isomorphism, we recall again that positive maps
correspond to block-positive matrices. We also note that completely positive maps correspond to
positive semi-definite matrices, and completely copositive maps correspond to
the partial transposes of positive semi-definite matrices. See \cite{ssz} for more details.
In this sense, whenever it does not lead to any confusion, we call a positive map
an entanglement witness if it detects entanglement, and vice versa.

Recall that a point $x$ in a convex set $C$ is said to be an interior point of $C$
if for each $y\in C$ there is $t>1$ such that $(1-t)y+tx\in C$.
A point $x$ is an interior point of $C$ in this sense if and only if it is a relative topological interior point
of $C$ with respect to the affine manifold generated by $C$.
A point of $C$ is said to be
a boundary point if it is not an interior point.
Note also that a point $x$ in a convex set $C$ determines a unique face of $C$ in which $x$ is an interior point.
See \cite{rock}. This is the smallest face of $C$ containing the point $x$.
For a positive map $\Phi$, we denote by $\mathbb P_\Phi$ the face of $\mathbb P_1$ in which
$\Phi$ is an interior point. When $\Phi$ is decomposable,
the corresponding face generated by $\Phi$ will be denoted by $\mathbb D_\Phi$.

It was shown in \cite{lew00} that an entanglement witness
$\Phi$ is optimal if and only if it satisfies the following condition:
\begin{itemize}
\item[(O1)]
For any completely positive map $\phi$ and $t>0$, the linear map $(1+t)\Phi-t\phi$ does not belong to $\mathbb P_1$.
\end{itemize}
With the change of the variable, this is equivalent to the following:
\begin{itemize}
\item[(O2)]
For any completely positive map $\phi$ and $t>1$, the linear map $(1-t)\phi+t\Phi$ does not belong to $\mathbb P_1$.
\end{itemize}
Geometrically, an entanglement witness $\Phi\in\mathbb P_1$ is optimal
if and only if any line segment from a completely positive map
$\phi$ to $\Phi$ cannot be extended inside of $\mathbb P_1$. The following fact was shown in \cite{sar}. We include a proof
for the convenience of readers.

\begin{theorem}{Proposition}\label{ew}
An entanglement witness $\Phi\in\mathbb P_1$ is optimal if and only
if there is no completely positive map in the face $\mathbb P_\Phi$.
\end{theorem}

\begin{proof}
Suppose that $\Phi$ is optimal. If there is a completely positive map in $\mathbb P_\Phi$ then we have a contradiction
to the condition (O2), since $\Phi$ is an interior point of $\mathbb P_\Phi$.

For the converse, suppose that $\Phi$ is not optimal.
Then there is a completely positive map $\phi$ and $t>1$ such that $\psi=(1-t)\phi+t\Phi$ belongs to $\mathbb P_1$.
Then $\Phi\in\mathbb P_\Phi$ is a nontrivial convex combination of $\phi$ and $\psi$. Since $\mathbb P_\Phi$
is a face, we conclude that the completely positive map $\phi$ lies in $\mathbb P_\Phi$.
\end{proof}

The faces of the cone $\mathbb D$ are determined \cite{kye_decom}
by pairs of subspaces of the space $M_{m\times n}$ of all $m\times n$ matrices, which is inner space isomorphic
to $\mathbb C^n\otimes\mathbb C^m$. Note that a product vector $\bar y\otimes x$ corresponds to
the rank one matrix $xy^*\in M_{m\times n}$ with this isomorphism,
where $x\in\mathbb C^n$ and $y\in\mathbb C^m$ are considered as column matrices.
More precisely, every face of $\mathbb D$ is of the form
$$
\sigma(D,E)=
\conv\{\phi_{V_i},\ \phi^{W_j}: V_i\in D,\ W_j\in E\}
$$
for a pair $(D,E)$ of subspaces of $M_{m\times n}$.
This pair $(D,E)$ is uniquely determined by a given
face $F$ under the condition
$$
D=\spa\{V:\phi_V\in F\},
\qquad
E=\spa\{W:\phi^W\in F\}.
$$
It should be noted that it is difficult in general
to determine if an arbitrary given pair of subspaces gives rise to a face of $\mathbb D$.
Actually, the notion of optimality turns out to be closely related with the question as to which pairs give rise to
faces of the cone $\mathbb D$.

Note that the completely copositive map $\phi^W$ is completely positive
if and only if $W$ is of rank one, and we have the relation
$\phi^{xy^*}=\phi_{\bar xy^*}$
in this case. Therefore, Proposition \ref{ew} tells us that if $\Phi$ is an optimal decomposable
entanglement witness then the face $\mathbb D\cap \mathbb P_\Phi$ of $\mathbb D$
must be of the form $\sigma(0,E_1)$ for a completely entangled subspace
$E_1$ of $M_{m\times n}$.
In particular, $\Phi$ must be completely copositive, and of the form
\begin{equation}\label{cocom}
\Phi=\phi^{W_1}+\phi^{W_2}+\cdots+\phi^{W_\nu},
\end{equation}
Since the relation $\mathbb D_\Phi\subset \mathbb D\cap \mathbb P_\Phi$ holds in general, we see that
$\mathbb D_\Phi$ is of the form $\sigma(0,E_2)$ for a subspace $E_2$ of $E_1$.
We also note that the map $\Phi$ of the form (\ref{cocom}) is an interior point of the convex set $\sigma(0,E_3)$
with $E_3=\spa\{W_1,\dots,W_\nu\}$ by
\cite{kye-cambridge}. Therefore, we conclude that $E_2=E_3$, and
$$
\mathbb D_\Phi=\sigma (0,E),\qquad {{\rm where}}\ E=\spa\{W_1,\dots,W_\nu\}.
$$
In this case, we say that $\Phi$ is {\sl supported} on the space $E=\spa\{W_1,W_2,\dots,W_\nu\}$.
Especially, we see that $E$ must be completely entangled and the convex set $\sigma(0,E)$ must be a face of $\mathbb D$,
whenever $\Phi$ in (\ref{cocom}) gives rise to an optimal entanglement witness.

Recall again that a completely copositive map $\Phi$ corresponds the partial transpose
$Q^\tau$ of a positive semi-definite matrices $Q$,
and the support of $\Phi$ in the above sense corresponds to the support of $Q$ itself.
Therefore, this recovers the necessary condition (B) for the optimality given in \cite{lew00},
as was mentioned in Introduction.

If $\Phi$ is an interior point of the convex cone $\mathbb P_1$ then the face $\mathbb P_\Phi$ is the
whole cone $\mathbb P_1$ itself. Therefore, Proposition \ref{ew} tells us that if $\Phi$ is optimal
then it must be on the boundary of the cone $\mathbb P_1$.
Note that the boundary structure
of a convex set is determined by their maximal faces. It is also known that every maximal face of the cone $\mathbb P_1$
is determined by a product vector, as was studied in \cite{kye_canad}.
More precisely, $\Phi$ lies on the boundary of the cone
$\mathbb P_1$ if and only if there is a rank one projector $(\bar y\otimes x)(\bar y\otimes x)^*\in \mathbb V_1$
on the product vector $\bar y\otimes x$
such that
$$
\langle (\bar y\otimes x)(\bar y\otimes x)^*,\Phi\rangle =0.
$$
If $\Phi=\phi^V$ is a completely copositive map then we have
$$
\langle (\bar y\otimes x)(\bar y\otimes x)^*,\phi^V\rangle=|(V\, |\, \bar y\otimes \bar x)|^2
$$
as in \cite{ha+kye_exposed},
where $(\ \, |\, \ )$ denotes the inner product of $M_{m\times n}$ which is linear in the first variable
and conjugate-linear in the second variable.
Therefore, if $\Phi$ is given by (\ref{cocom}) then we have the following relation
$$
\langle (\bar y\otimes x)(\bar y\otimes x)^*,\Phi\rangle
=\sum_{i=1}^\nu |(W_i\, |\, \bar y\otimes \bar x)|^2.
$$
This shows that $E^\perp$ has a product vector.
We summarize as follows:

\begin{theorem}{Theorem} \label{main}
Let $\Phi$ be a completely copositive map
supported on the subspace $E$ of $M_{m\times n}=\mathbb C^n\otimes\mathbb C^m$. If $\Phi$ gives rise to
an optimal entanglement witness then we have the following:
\begin{enumerate}
\item[{\rm (i)}]
$E$ is completely entangled.
\item[{\rm (ii)}]
$E^\perp$ has a product vector.
\item[{\rm (iii)}]
The convex set $\sigma(0,E)$ is a face of $\mathbb D$.
\end{enumerate}
\end{theorem}

We note that the condition (ii) seems to be known to specialists as it was pointed out by one of the referee,
even though it is not mentioned explicitly in the literature. See \cite{2} and \cite{1}.
Note also that the condition (iii) is equivalent to saying that $\sigma(E,0)$ is a face of $\mathbb D$
by the symmetry. This is exactly
the condition (D) mentioned in Introduction.

\section{Examples}

In the case of $m=2$, it was shown in \cite{aug} that
if $E$ is completely entangled
then $E^\perp$ is spanned by product vectors, and the converse of Theorem \ref{main} holds.
In this case, $D(0,E)$ is a face of $\mathbb D$ which is actually exposed by separable states
with respect to the duality (\ref{ddual}) by \cite{choi-kye-exp}.
In the case of $m=n=3$, it is well known that
the maximum dimension of completely entangled subspaces is $4$. Therefore, when $m=2$ or $m=n=3$
the condition (i) of Theorem \ref{main} actually implies the condition (ii).

Consider a $3\otimes 3$ PPT state $A$ with
$A=A^\tau$ which is supported on a $4$-dimensional completely entangled subspace $E$, as it was given in
\cite{ha+kye}. Since the $9\times 9$ matrix $A$ represents both a completely positive map and a completely copositive map,
this is a common interior point of the different convex subsets $\sigma(E,0)$ and $\sigma(0,E)$ of the cone
$\mathbb D$, and so we conclude that $\sigma(0,E)$ is not a face of the cone $\mathbb D$. Hence, any interior point
of $\sigma(0,E)$ which is not completely positive gives us an example of a non-optimal entanglement witness
arising from a completely copositive map supported on the completely entangled subspace $E$.

To get an explicit example in terms of positive semi-definite
block-matrices, we begin with
$$
A_0=
\left(\begin{array}{ccccccccc}
1&\cdot&\cdot & \cdot&1&\cdot & \cdot&\cdot&1\\
\cdot&\lambda^2&\cdot & 1&\cdot&\cdot & \cdot&\cdot&\cdot\\
\cdot&\cdot&\mu^2 & \cdot&\cdot&\cdot & 1&\cdot&\cdot\\
\cdot&1&\cdot & \mu^2&\cdot&\cdot & \cdot&\cdot&\cdot\\
1&\cdot&\cdot & \cdot&1&\cdot & \cdot&\cdot&1\\
\cdot&\cdot&\cdot & \cdot&\cdot&\lambda^2 & \cdot&1&\cdot\\
\cdot&\cdot&1 & \cdot&\cdot&\cdot & \lambda^2&\cdot&\cdot\\
\cdot&\cdot&\cdot & \cdot&\cdot&1 & \cdot&\mu^2&\cdot\\
1&\cdot&\cdot & \cdot&1&\cdot & \cdot&\cdot&1
\end{array}\right),\
A_1=
\left(\begin{array}{ccccccccc}
1&\cdot&\cdot & \cdot&1&\cdot & \cdot&\cdot&1\\
\cdot&\cdot&\cdot & \cdot&\cdot&\cdot & \cdot&\cdot&\cdot\\
\cdot&\cdot&\cdot & \cdot&\cdot&\cdot & \cdot&\cdot&\cdot\\
\cdot&\cdot&\cdot & \cdot&\cdot&\cdot & \cdot&\cdot&\cdot\\
1&\cdot&\cdot & \cdot&1&\cdot & \cdot&\cdot&1\\
\cdot&\cdot&\cdot & \cdot&\cdot&\cdot & \cdot&\cdot&\cdot\\
\cdot&\cdot&\cdot & \cdot&\cdot&\cdot & \cdot&\cdot&\cdot\\
\cdot&\cdot&\cdot & \cdot&\cdot&\cdot & \cdot&\cdot&\cdot\\
1&\cdot&\cdot & \cdot&1&\cdot & \cdot&\cdot&1
\end{array}\right),
$$
where $\lambda\mu=1$ and $\lambda\neq 1$.
Note that the matrix $A_0$ is just the matrix $A$ in
\cite{ha+kye}. It is supported
on the
subspace $E$ spanned by
$$
|00\ran +|11\ran+|22\ran,\qquad \lambda|01\ran+\mu|10\ran,\qquad \lambda|12\ran+\mu|21\ran,\qquad  \lambda|20\ran+\mu|02\ran,
$$
which is completely entangled.
We also note that  $A_1^\tau$ represents a positive map in $\sigma(0,E)$ which is not completely positive.
We see that the partial transposes $A_t^\tau$ of any matrices $A_t$ on the open line segment between
two matrices give rise to the required examples:  $A_t$
is positive semi-definite matrix supported on the
completely entangled subspace $E$. But, $A_t^\tau$ is not optimal since $A_0^\tau=A_0$ is positive semi-definite.
See the recent papers \cite{hhms} and \cite{sko} for more examples of PPT states
which are symmetric under the partial transpose.

In the case of $m=3$ and $n=4$,
the maximum dimension of completely entangled subspaces is $(3-1)\times (4-1)=6$. Furthermore,
generic $6$-dimensional subspaces of $M_{3\times 4}$ are completely entangled. Therefore,
it might be very reasonable to expect that there is a $6$-dimensional completely entangled subspace of
$M_{3\times 4}$ whose orthogonal complement is also completely entangled. Indeed, numerical results in
\cite{lein} indicate that this is the case. They produced $3\otimes 4$ PPT states
whose images and kernels are $6$-dimensional completely entangled subspaces.
We note that concrete examples were found recently by
Augusiak et al. \cite{asl} and
Skowronek \cite{sko-exam}.
More generally, Young-Hoon Kiem informed the author that if
$$
mn-(m-1)(n-1)\le k\le (m-1)(n-1)
$$
then $k$-dimensional subspaces of $M_{m\times n}$ are generically completely entangled subspaces with
the completely entangled orthogonal complements.

Characterizing completely entangled subspace $D$ satisfying the condition (iii) of Theorem \ref{main}
would be the first step to understand the whole facial structures of the cone $\mathbb D$: How the two convex
cones consisting of completely positive maps and completely copositive maps respectively are located each other to
generate the whole convex cone $\mathbb D$. See \cite{choi-kye-exp} for the partial results in the case of $m=2$.
It would be also interesting to know if there is a face of the form $\sigma(0,E)$ which not exposed. Note that every known example
of unexposed face of $\mathbb D$ is of the form $\sigma(D,E)$ with nonzero $D$ and $E$. See \cite{byeon-kye}
and \cite{choi-kye-exp}.

It would be interesting to know if the converse of Theorem \ref{main} holds or not.
If we add to the three conditions in Theorem \ref{main}
one more condition that the two faces $\mathbb D\cap \mathbb P_\Phi$ and $\mathbb D_\Phi$ coincide,
then we see that $\Phi$ is optimal by Proposition \ref{ew}.
To find conditions under which the relation $\mathbb D\cap \mathbb P_\Phi=\mathbb D_\Phi$ holds in the general situation
seems to be an important step to understand in what way the cone $\mathbb D$ sits down in the bigger cone $\mathbb P_1$.

The author is grateful to Kil-Chan Ha for valuable comments on the first draft of this note.
He is also grateful to Remigiusz Augusiak, Young-Hoon Kiem and \L ukasz Skowronek for valuable discussions on the topics.
Finally, he is grateful to one of the referees for bringing his attention to the papers \cite{2}, \cite{1}
and \cite{sar}.

\end{document}